\documentclass[twocolumn]{apie}
\usepackage{amsmath,amsfonts,amsbsy}
\usepackage{bm}
\usepackage{algorithm}
\usepackage{algorithmic}
\usepackage{array}
\usepackage{amsthm}
\usepackage{hyperref}
\usepackage{tikz}
\usepackage{subcaption}
\usetikzlibrary{arrows.meta}
\usetikzlibrary{shapes.geometric, arrows}
\usepackage{circuitikz}
\usetikzlibrary{shapes}
\usepackage{tikz}
\usetikzlibrary{positioning,shapes,arrows,arrows.meta}
\tikzstyle{startstop} = [draw, rounded rectangle, text centered, draw=black,thick]
\tikzstyle{io} = [draw=black,thick,trapezium, trapezium left angle=70, trapezium right angle=110,text centered]
\tikzstyle{process} = [rectangle, text centered, draw=black,thick]
\tikzstyle{decision} = [diamond, text centered, draw=black,thick]
\tikzstyle{arrow} = [-{Stealth[scale=1.2]},rounded corners,thick]
\usepackage{textcomp}
\usepackage{lipsum}
\usepackage{stfloats}
\usepackage{tabularx,booktabs}
\usepackage{url}
\usepackage{diagbox}
\usepackage{verbatim}
\usepackage{multirow}
\usepackage{graphicx}
\usepackage{float}
\usepackage{multirow}
\usepackage{algorithmic}
\usepackage{amsmath}
\usepackage{amsthm}
\usepackage{array}
\usepackage{url}
\usepackage{makecell}

\newtheorem{proposition}{Proposition}
\usepackage{threeparttable}
\usepackage{xcolor}
\usepackage[english]{babel}
\newtheorem{theorem}{Theorem}[section]

\newtheorem{lemma}[theorem]{Lemma}

\newtheorem{Assumption}[theorem]{Assumption}
\newtheorem{Definition}[theorem]{Definition}
\usepackage{amsmath}

\usepackage{url}
\newcolumntype{b}{X}
\newcolumntype{s}{>{\hsize=.5\hsize}X}

\title{Carbon Emission Flow Tracing: \\Fast Algorithm and  California Grid Study}

\author[1]{Yuqing Shen}
\author[1]{Yuanyuan Shi}
\author[2]{Daniel Kirschen}
\author[3]{Yize Chen}

\affiliation[1]{University of California San Diego}
\affiliation[2]{University of Washington}
\affiliation[3]{University of Alberta}

\definecolor{verylightgray}{rgb}{0.9,0.9,0.9}
\definecolor{lightblue}{rgb}{0.733,0.875,1.0}
\definecolor{metablue}{rgb}{0, 0.392, 0.898}
\definecolor{forestgreen}{rgb}{0.208,0.667,0.235}
\definecolor{verylightblue}{rgb}{0.839,0.925,1.0}
\definecolor{veryverylightblue}{rgb}{0.92,0.96,1.0}

\abstract{
Power systems decarbonization are at the focal point of the clean energy transition. While system operators and utility companies increasingly publicize system-level carbon emission information, it remains unclear how emissions from individual generators are transported through the grid and how they impact electricity users at specific locations.
This paper presents a novel and computationally efficient approach for exact quantification of nodal average and marginal carbon emission rates, applicable to both AC and DC optimal power flow problems. 
The approach leverages graph-based topological sorting and directed cycle removal techniques, applied to directed graphs formed by generation dispatch and optimal power flow solutions. Our proposed algorithm efficiently identifies each generator’s contribution to each load node, capturing how emissions are spatially distributed under varying system conditions.
To validate its effectiveness and reveal locational and temporal emission patterns in real world, we simulate the 8,870-bus realistic California grid using the real CAISO data and CATS model~\citep{CATS2024}. Based on year-long hourly data on nodal loads and renewable generation—obtained or estimated from CAISO public data, our method accurately estimates power flow conditions, generation mixes, system total emissions, and delivers fine-grained spatiotemporal emission analysis for every California county.
Both our algorithm and the California study are open-sourced\footnote{\url{https://github.com/yuqing5/Carbon-Tracker-California}}, providing a foundation for future research on grid emissions, planning, operations, and energy policy.
}
\date{\today}
\correspondence{Yuanyuan Shi \email{yyshi@ucsd.edu} and \\Yize Chen \email{yize.chen@ualberta.ca}}

\begin{document}

\maketitle

\section{Introduction}
Achieving low-carbon or net-zero emissions in power systems is a critical priority for tackling climate change and enabling a sustainable energy future.
In recent years, many countries, states, and companies have set ambitious decarbonization goals for the next two to three decades~\citep{rockstrom2017roadmap,papadis2020challenges,CAL_Carbon2050,Google_Carbon2030}. As sectors such as transportation, heating, and computing become increasingly electrified—with data centers consuming rapidly growing amounts of power—the demand for clean electricity continues to rise~\citep{xie2024sustainable}. Accurate carbon emission quantification serves a wide range of important applications: supporting emissions accounting, informing operational assessments and green tariff design, and guiding investment decisions on the siting and sizing of carbon-free resources~\citep{ling2024comprehensive}. Reflecting growing public and institutional interest, a few independent system operators have begun publishing system-level emission data, such as total greenhouse gas emissions from CAISO~\citep{caisoemission} and marginal emission rates from PJM~\citep{PJM}.

Fine-grained, high-quality emission data are critical for decarbonizing the power grid by guiding generation dispatch, load shifting and emission-oriented planning~\citep{cheng2022carbon}. Recent greenhouse gas (GHG) Protocol's proposed revisions to Scope 2 guidance also highlights the requirements of such emission tracking capabilities under new carbon accounting standards~\citep{GHG2025}. 
Geographically distributed generators vary widely in technologies and emissions, and the electricity they produce is transported through the interconnected transmission and distribution networks to end users. Consequently, electricity consumed at a given location can reflect a different mix of sources~\citep{kirschen1997contributions} and emission rates~\citep{cho2025pglib,gorka2025electricityemissions}. This makes it critical to understand \emph{when} and \emph{where} emissions occur to inform operational decisions and decarbonization strategies. More recently, large electricity consumers such as manufacturing and IT companies face challenges in quantifying Scope 2 emissions—indirect emissions from purchased electricity—due to the lack of time- and location-specific data tied to electricity usage~\citep{carbonpresentation}. High-resolution emission data are increasingly needed by policymakers, grid operators, and electricity end users who aim to track and reduce carbon impacts. These stakeholders can leverage high-fidelity spatiotemporal emission data to evaluate climate policies~\citep{li2013comparing}, provide the right incentives~\citep{li2022combined}, and design emission-aware planning and operational models~\citep{lindberg2021guide, shaocarbon} that support a just and efficient clean energy transition. 

However, analyzing fine-grained spatiotemporal emission patterns remains challenging due to the inherent complexity of the power grid. The network structure, time-varying load and generation profiles, and physical limits on line flows make it difficult to attribute emissions from generators to specific locations and times. As a result, most estimates are coarse, system-level aggregates based on total energy mix, with limited spatial resolution. This limitation has real consequences. Some emission reduction strategies can lead to unintended effects—for instance, \citep{gorka2024electricityemissions} shows that load reduction at certain nodes may actually increase overall system emissions. 

Two key metrics in this context are the \emph{locational average emission} rate (LAE) capturing total emissions per MWh over the entire electricity usage, and the \emph{locational marginal emission} rate (LME) reflecting the emissions impact of a small change in local demand. The former, LAE, is particularly useful for quantifying total emissions of end users, supporting Scope 2 emission reporting and sustainability goal tracking~\citep{GHG2025}. The latter parallels the concept of locational marginal price (LMP) and captures the sensitivity of emissions to local demand changes, making it well-suited for demand-side response, load shifting, and carbon reduction incentives.

\begin{figure}[th]
    \centering
    \includegraphics[width=0.99\linewidth]{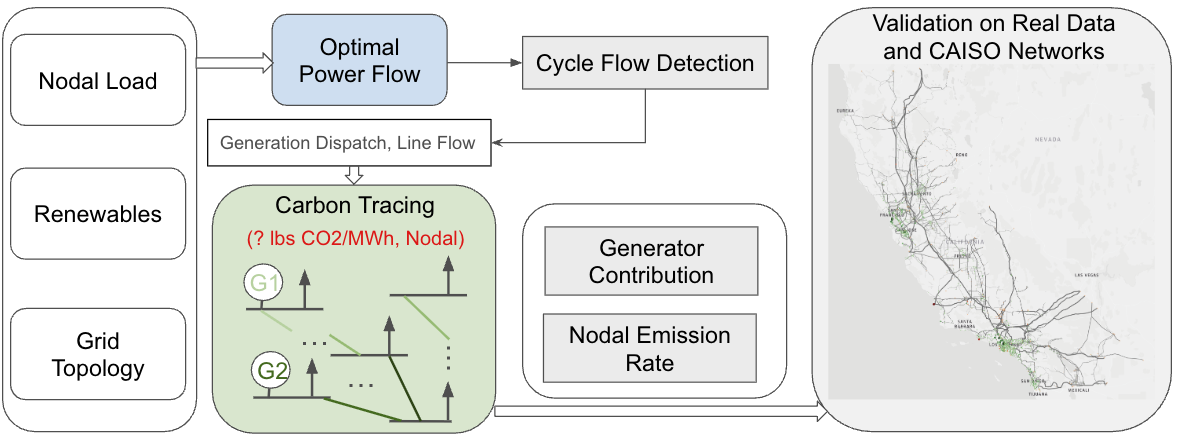}
    \caption{Proposed framework for carbon tracing, fully compatible with ACOPF and applicable for nodal average and marginal emission rate.}
    \label{fig:flowchart}
\end{figure}
In this work, we develop a novel carbon tracing algorithm that supports both AC and DC power flow models and efficiently computes LAE and LME.
Our method takes  the grid topology, nodal demand, and renewable and conventional generation as input, and computes the exact contribution of each generator to each bus and transmission line. 
One key challenge in carbon tracing arises from cycle flows in the directed power flow graph, which can result from modeling artifacts, ACOPF non-convexity, or topological errors. We address this challenge by identifying and eliminating cycles using strongly connected components (SCCs), collapsing each SCC into a supernode. This transformation yields a directed acyclic graph, enabling an order of magnitude  faster computation of nodal emission rates with linear complexity. We validate our approach through a large-scale case study using the California Test System (CATS) model~\citep{CATS2024}. CATS provides geographically accurate, realistic-yet-public grid topology data and is compatible with CAISO generation and demand records. Using approximated nodal demand and renewable generation, we solve the optimal power flow (OPF) problem to estimate power dispatch decisions, line power flows, and compute nodal emission rates. To validate the fidelity of our approach, we aggregate our nodal estimates and compare them to CAISO's reported system-wide emissions. Across a full year of hourly data, our method achieves total emission estimates within 6\% of CAISO’s published values. Fig.~\ref{fig:flowchart} illustrates the overall procedure. The major contributions are summarized as follows:

\begin{itemize}
\item We develop a novel and computationally efficient carbon tracing approach, which is built upon physical power flow, and can be applied to both locational average and marginal emission rate calculations.
\item Proposed method is fully compatible with full AC power flow, which reduces SCCs and avoids matrix inversions. The method and findings presented here offer a quantitative foundation for practices on locational emission analysis, including renewable impact assessment, siting of low-carbon resources, and generation expansion planning. Together, these tools can inform actionable strategies for emission reduction and sustainable energy transitions.
 \item To the best of our knowledge, we provide the first study to conduct nodal emission analysis on a realistic ISO-level power system, and can provide hourly-matched and highly accessible carbon tracing. Results indicate there could be over 100 times difference in terms of county-level emission rate, highlighting substantial spatial disparities driven by generation mix and power flow. 
\end{itemize}

\subsection{Related Works}

\emph{Emission Metrics for Power Grids:} 
The recent surge in global computing demand—driven by cloud services and hyperscale data centers—has increased pressure on electric grids and carbon footprints. To mitigate the environmental impact, recent efforts include carbon-intelligent computing that shifts workloads to cleaner hours or regions~\citep{radovanovic2022carbon}, as well as emerging solutions that align computing with grid conditions to enable zero-carbon cloud operations~\citep{utilityarticle}.
At the same time, there is ongoing debate about how best to quantify electricity-related emissions~\citep{cho2025pglib,gorka2025electricityemissions}. System-level emission metrics reflect macro-scale net emissions but largely overlook geographic variations. To capture location-based emission rates, prior works have used solar irradiance and meteorological data (e.g., temperature and wind speed) to estimate renewable generation and utilization across specific regions~\citep{zhang2025review}. These location-specific emission rate metrics are also compatible with carbon pricing mechanisms~\citep{li2022combined}, which are designed to incentivize greenhouse gas reductions on both the generation and demand sides~\citep{cheng2019low}. However, these methods primarily provide regional-level estimates of grid-side carbon intensity~\citep{maji2022carboncast}.

More recently, LAEs and LMEs have gained significant attention as effective metrics for enabling load shifting and emission reduction in energy-intensive applications~\citep{lindberg2021guide}. On the one hand, model-based approaches explicitly calculate such metrics using dispatch models. For instance, to calculate LMEs, implicit differentiation  take dynamic constraints like storage and ramp-constraint generators into consideration~\citep{valenzuela2023dynamic}. While emission sensitivity analysis for LMEs and LAEs based on market clearing model are also studied~\citep{lu2024market, kang2015carbon}. Model-based approaches in the literature often require matrix inversion, which is both computationally expensive and applicable mainly for linearized power flow. On the other hand, data-driven LME estimation methods can face generalization challenges across varying topology and load~\citep{mayes2024using,shaocarbon}.

While most prior efforts address the quantification of LMEs and LAEs  separately, our approach provides a unified framework based on physical power flow tracing. 
Last but not least, methods reviewed above—including ours—fall under the category of location-based emission accounting. In contrast, market-based approaches~\citep{hulshof2019performance} estimate carbon emissions based on contractual arrangements for electricity procurement, such as Renewable Energy Certificates (RECs) or Power Purchase Agreements (PPAs), which is beyond the scope of this work.

\begin{figure*}[h]
    \centering
    \includegraphics[width=0.99\linewidth]{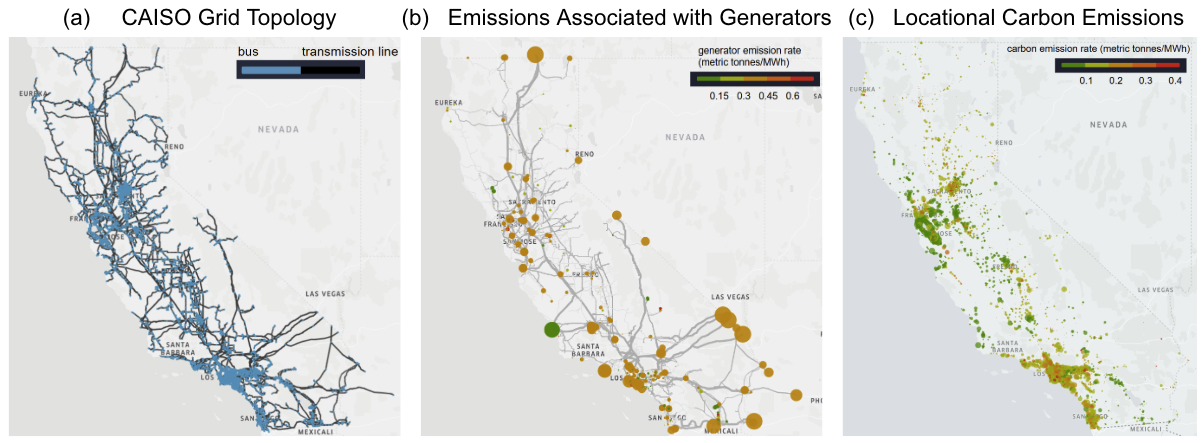}
    \caption{Our proposed module can not only find locational emission rates in fast algorithms, but also provide visualization of: (a). Transmission grid topology; (b). Emissions associated with different generators, where the node size denotes the generation level; (c). Nodal average carbon emission rates, where the node size denotes the demand level. Shown results are based on CATS model for California grid on June 17th, 2019.}
    \label{fig:large}
\end{figure*}

\emph{Carbon Tracing and Carbon Accounting:}
Pioneering works have established country-wide, real-time emission accounting frameworks. However, many do not model the exact physical power and carbon flows~\citep{electricitymaps}, or are limited to estimating emissions at the balancing-authority or broader regional levels. For example, US-wide and European grid emission tracking are explored in~\citep{de2019tracking} and~\citep{tranberg2019real}, respectively, but these studies rely on macro-level flow models rather than detailed transmission-level flows. Other methods~\citep{chen2024carbon, kang2015carbon} apply proportional sharing principles to allocate generator output and use matrix inversion to estimate the share of each generator’s power—and its associated emissions—attributed to each load.

Our work is most closely related to the network-based locational emission tracing~\citep{kang2015carbon, chen2024carbon}, which is recently implemented as the CarbonFLow\texttrademark \; tool by MISO in US~\citep{MISO2025}. While such works assume a known, invertible power distribution factor matrix, limiting the use cases to DCOPF solutions without cycle flows with at least $\mathcal{O}(N^3)$ complexity. Built upon our previous tree-based method by running depth first search (DFS) to trace back each generator's emission contribution and power proportions at lines and buses \citep{chen2024contributions}, this work is fully compatible with full ACOPF with much reduced computation complexity.

\section{Preliminaries}
\subsection{Overview of the Proposed Framework}
In this work, we aim to quantify each node’s carbon emission rate from power usage. At a given time snapshot, our approach takes the load vector and solves a general optimal power flow (OPF) problem using standard formulations and solvers. Once a feasible solution for generation dispatch and line flows is obtained, our algorithm traces each generator’s contribution to each demand node. To achieve so, we implement a tree-based algorithm recursively, until all generators are enumerated. This allows the power from every generator to be allocated and attributed to each demand node. 

The full workflow is illustrated in Figure~\ref{fig:flowchart}, and we particularly evaluate our algorithm for the whole California grid based on the open-sourced CATS model~\citep{CATS2024}, which preserves transmission line accuracy and reflects grid fidelity at the same time. We begin with grid topology, renewable generation and line parameters acquired from CAISO and CATS repositories. Nodal load and renewables data are then fed into the OPF solver to obtain corresponding generation and line power flow values (Section~\ref{sec:opf}). Once OPF instances are solved with feasibility and optimality guarantee, our carbon tracing algorithm computes the exact nodal power supply contributed by each generator with cycle flows considered (Sections~\ref{sec:carbon_tracing1} and~\ref{sec:carbon_tracing2}). Locational emission rates are then derived based on EPA~\citep{EPA1} and EIA~\citep{EIA1} actual generator's emission rate datasets (see Table \ref{table:emission_rate} in Appendix).  

\subsection{OPF and Power Dispatch}
\label{sec:opf}
We consider a connected power network $G=(\mathcal{N}, \mathcal{L})$ with $N$ nodes and $L$ power lines.  $\mathcal{N}$ denotes the set of all nodes, and $\mathcal{L}$ the set of lines.  Also let $\mathcal{G}$ denote the set of generators. At each timestep, given load vector  $\mathbf{p}^d \in \mathbb{R}^N$, we assume the system operator solves the OPF problem and finds the power dispatch $\mathbf{p}^g \in \mathbb{R}^{K}$ for $K$ generators. The corresponding line flow $\mathbf{p}_l\in \mathbb{R}^{L}$ is also solved. Without loss of generality, for line pair $(i,j)\in \mathcal{L}$, we use $p_{ij}$ and $q_{ij}$ to denote the directed  active and reactive line flow from node $i$ to node $j$. We use $\mathcal{N}^+_i, \mathcal{N}^-_i$ to denote the set of neighbor nodes that send power to and receive power from node $i$, respectively.
\begin{subequations}
\label{equ:opf}
\begin{align}
\min_{\{V_i, \theta_i, p_{i}^g, q_{i}^g\}} \quad \; & \sum_{i \in \mathcal{G}} C_i(p_{i}^g) \label{equ:obj}\\
\textit{subject to: \quad }& \notag \\
 p_{i}^g - p_{i}^d = &\sum_{j \in \mathcal{N}} V_i V_j (G_{ij} \cos\theta_{ij}  + B_{ij} \sin\theta_{ij})\label{equ:active}\\
 q_{i}^g - q_{i}^d = & \sum_{j \in \mathcal{N}} V_i V_j (G_{ij} \sin\theta_{ij} - B_{ij} \cos\theta_{ij})\label{equ:reactive}\\
P_{G_i}^{\min} \leq & p_{i}^g \leq P_{G_i}^{\max}\label{equ:Plimit}\\
Q_{G_i}^{\min} \leq & q_{i}^g \leq Q_{G_i}^{\max}\label{equ:Qlimit}\\
V_i^{\min} \leq & |V_i| \leq V_i^{\max}\label{equ:Vlimit}\\
|S_{ij}| \leq & S_{ij}^{\max}\label{equ:Slimit}
\end{align}
\end{subequations}

In the formulation above, we use $C_i(\cdot)$ in \eqref{equ:obj} to denote the cost function of generator $i$; $\theta_i$ represents the phase angle at bus $i$. $p_{i}^{g}$ and $p_{i}^{d}$ represent the power generation and demand at node $i$. Equations \eqref{equ:active}–\eqref{equ:Slimit} impose constraints on nodal active and reactive power balance, generator limits, voltage magnitude, and line flows. By solving this optimization problem, we obtain the power generation and line flow that minimize the total cost.
Note that we consider the general form of ACOPF here, although our method is also compatible with alternative formulations—including those involving energy storage, multi-step optimization, and DCOPF—as long as power dispatch and line flows are solved. This is a distinguishing feature of our work compared to previous studies~\citep{MISO2025, kang2015carbon, chen2024carbon}, as detailed in later sections.

Once OPF problem \eqref{equ:opf} is solved, the nodal generations $\mathbf{p}^g$ and line flow $\mathbf{p}_l$ are revealed to the system operators. We denote the directed graph given by the OPF solution as $G^{\text{OPF}}$. Let $\gamma_k$ denote the carbon emissions rate of generator $k$, then the total emissions associated with this generator can be quantified as $\gamma_k \cdot p_k^g$.  We are interested in quantifying how much emissions are transported from generation to consumption through solved power flows $p_{ij}$ with emission rate $\gamma_{ij}$.

\subsection{Flow-Based Locational Emission Accounting}

To model how emissions are carried over by power flow, we follow two fundamental principles in our carbon tracing mechanism: 1) At each node $i$, there is a conservation of aggregated emissions, so that net carbon inflows equal the net carbon outflows~\footnote{We term the emissions associated with power inflow and nodal generations as net carbon inflows; emissions associated with power outflow and nodal consumptions as net carbon outflows}; 2) At each node and line, there is a perfect mix of power flow, such that emissions quantity are proportional to power flow quantity, e.g., in Fig. \ref{fig:toy_example} the average emission rate at load bus 3 is $\delta(p_3^d)=\frac{\gamma_{13} p_{13}+\gamma_{23} p_{23}}{p_{13}+p_{23}}$ based on the power flow. In addition, we also make the following assumption regarding the proportional share of power flow at each node, which is also widely used in previous literature~\citep{kirschen1997contributions, chen2024contributions, kang2015carbon}: 

\begin{Assumption}
\label{assumption}
For any node $i$, if the proportion of the inflow which can be traced to generator $k$ is $\alpha_{i}(p_k^g)$, then the proportion of the outflow which can be traced to generator $k$ is also $\alpha_i(p_k^g)$.
\end{Assumption}
Assumption \ref{assumption} provides a principle to proportionally allocate the inflow power and further the inflow carbon emissions by nodal demand and output power flow. 

Our goal is to find the carbon emissions associated with each load node. We want to compute the LAE rate $\delta(p_i^d)$ and LME rate $\mu(p_i^d)$. Both rates have a unit of metric tonnes CO$_2$/MWh. 
To achieve this goal, we find it possible to follow the physical interpretations of average carbon emission rate $\delta(p_i^d)$ via taking the division of nodal-level net emission by power demand $p_i^d$. The key for computing the total carbon emission is to find generator $k$'s contribution $p_i^d(p_k^g)$ to supply $p_i^d$ in MW\footnote{Throughout the paper, we use $p_i(x)$ to denote the power contributed by $x$ to node $i$, where $x$ denotes either generator or inflow. Similar definitions hold for line flow $f_{l}(x)$}, then we can compute the total emission associated with $p_i^d$ as $e(p_i^d)=\sum_{k=1}^{K} \gamma_k \cdot p_i^d(p_k^g)$. Thus, each node's average carbon emission rate (with $p_i^d>0$) is computed as,
\begin{equation}
\label{equ:emission}
    \delta(p_i^d)=\frac{\sum_{k=1}^{K} \gamma_k \cdot p_i^d(p_k^g)}{p_i^d}.
\end{equation}

In a similar vein, we can also define nodal marginal emission rate by locally perturbing the nodal load $p_i^d$: 

\begin{equation}
    \mu(p_i^d)=\frac{\partial ( \sum_{k=1}^{K} \gamma_k \cdot p_k^g ) }{\partial p_i^d}.
\end{equation}

\begin{figure}[h]
    \centering
    \includegraphics[width=0.95\linewidth]{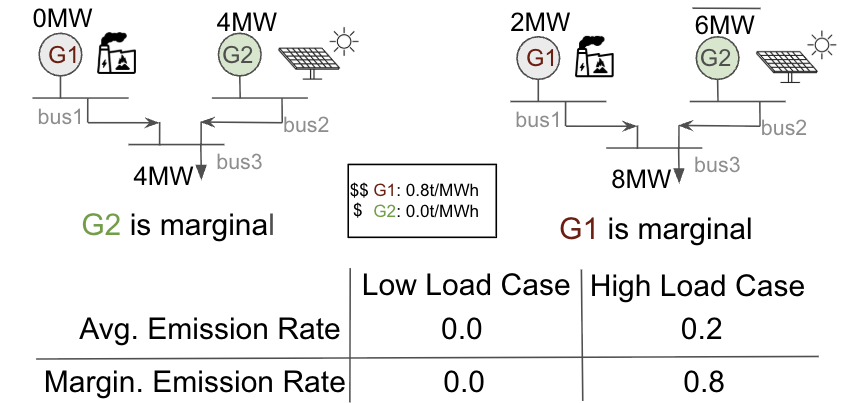}
    \caption{A 2-Generator, 1-Load illustration of average and marginal emission rate. }
    \label{fig:toy_example}
\end{figure}

In Fig. \ref{fig:toy_example}, we show a 3-bus toy example to illustrate how locational emission rates are quantified and varied due to load and generation conditions. In this simplified case, we assume the renewable generation are much cheaper but with fluctuating generations. When the load is small (left), cheaper renewables from Bus 2 are dispatched with $0$ emissions. When the renewable generation is binding at high load times, both Bus 3's average and marginal emission rate are changed due to traditional generator's output. Once power dispatch is determined based on \eqref{equ:opf}, the proposed method will firstly determine each generator's power contribution to each load and then trace generator's contribution to nodal and line flow emissions.

With known $\delta(p_i^d)$, one can also easily verify and calculate system's total emission $E=\sum_i \delta(p_i^d)\cdot p_i^d$. In next Section, we will describe our design for calculating $p_i^d(p_k^g)$ throughout the whole directed network, which is compatible with both ACOPF and DCOPF power dispatch solutions.

\section{Generator's Nodal Emission Contribution}
\subsection{Carbon Tracing Algorithm}
\label{sec:carbon_tracing1}
To determine the carbon emission intensity at each bus, once the OPF problem \eqref{equ:opf} is solved, and both direction and magnitude of line flows $\mathbf{p}_l, \; l\in \mathcal{L}$ are determined, we need to find the proportion of load demand at each bus that is supplied by each generator. If the bus is close to carbon-intensive coal generator and not connect to any renewable generators, it has high carbon emission rate. To determine the carbon intensity, we need to know the inflow of each bus and proportion of a generator contributed to each particular bus. So we construct the line power proportion matrix and the bus power proportion matrix as follows.

As described beforehand, the key to recover $\delta(p_i^d)$ is to find how much power is contributed by each generator $k$. Then it suffices to know every entry of the following Generator Node Distribution Factor (GNDF): 

\begin{Definition} 
    For every bus $i$, denote the Generator Node Distribution Factor (GNDF) as generator $k$'s power allocated to bus $i$ $p_i^d(p_k^g)$  with respect to bus $i$'s total power demand $p_i^d$, which can be computed as $\text{GNDF}_k^i=\frac{p_i^d(p_k^g)}{p_i^d}$.  
    For the whole network, The  GNDF matrix is defined as $\text{GNDF}=\{\text{GNDF}_k^i\}\in \mathbb{R}^{K\times N}$, with $\sum_k \text{GNDF}_k^i=1, \; \forall i$.
\end{Definition}

Similarly, for each line $ij$ we can define how much power is supported by generator $k$ using a matrix termed Generator Line Distribution Factor (GLDF):  

\begin{Definition}  
    For every line $l \in \mathcal{L}$, denote the Generator Line Distribution Factor (GLDF) of generator $k$'s power allocated to line $l$ $p_l(p_k^g)$ with respect to line $l$'s power flow $f_l$, which can be computed as $\text{GLDF}_k^i=\frac{p_l(p_k^g)}{f_l}$. 
    For the whole network, The GLDF matrix is defined as $\text{GLDF}=\{\text{GLDF}_k^i\}\in \mathbb{R}^{K\times L}$, with $\sum_k \text{GLDF}_k^l=1,\;  \forall l$.
\end{Definition}

Now we are ready to infer every entry of $\text{GNDF}$ and $\text{GLDF}$ in an iterative approach. At bus $i$, denote the set of input and output power lines as $\mathcal{L}_{\text{in}}^i$ and  $\mathcal{L}_{\text{out}}^i$ respectively. Then the following relationship holds for entry $\text{GNDF}_k^i$: 
\begin{equation}
\label{equ:GNDF}
\text{GNDF}_k^i=\frac{\sum_{l\in \mathcal{L}_{\text{in}}^i} \text{GLDF}_k^l\cdot f_l}{\sum_{l\in \mathcal{L}_{\text{in}}^i}f_l} = \frac{\sum_{l\in \mathcal{L}_{\text{out}}^i} \text{GLDF}_k^l\cdot f_l}{\sum_{l\in \mathcal{L}_{\text{out}}^i} f_l}.
\end{equation}

Once at node $i$, the value of $\text{GNDF}_k^i$ is calculated, the outflow line power contributed by the generator $k$ can be also determined. This is because given the current bus $i$, for $l \in \mathcal{L}_{\text{out}}^i$, the line power proportion from generator $k$ is equal to the bus power proportion from generator $k$ at bus $i$. 
\begin{equation}
\label{equ:GLDF}
    \text{GLDF}_k^l = \text{GNDF}_k^i, \; \forall l \in \mathcal{L}_{\text{out}}^i
\end{equation}



Equ \eqref{equ:GNDF}-\eqref{equ:GLDF} indicate that to calculate $\delta(p_i^d)$, we can take an iterative approach by traversing through all the nodes and determine all entries in $\text{GLDF}$ and $\text{GNDF}$ matrices. At a given bus $i$, the nodal $\text{GNDF}_k^i$ can be calculated based on known $\text{GLDF}$ from all neighboring upstream nodes. While for all neighboring downstream line flows, $\text{GLDF}_k^l$ can be updated once $\text{GNDF}_k^i$ is known.

Fig. \ref{fig:SCC} provides an illustrative example from the full
ACOPF of the CAISO network based on the CATS model. Snippet of the whole network is shown for visualization purpose. In Fig. \ref{fig:SCC}(b), $n1$ is connected to the generator $G2$ and has no other inflow power lines (noting the direction of the solved power flow, $n1$ only has outflows but not inflows), thus $\text{GNDF}_{G2}^{n1}$ = 1, meaning bus $n1$ is $100\%$ supplied by $G2$. $n1$ transfers power to bus $n3$ and bus $n4$. Based on Equ \eqref{equ:GLDF}, line $f_{n1,n4}$ and $f_{n1,n3}$ are also $100\%$ supplied by $G2$, leading to $\text{GLDF}_{G2}^{n1,n3} = \text{GLDF}_{G2}^{n1,n4} = \text{GNDF}_{G2}^{n1}= 1$. While at bus $n4$, power mix happens, requiring to know the quantity and power mix of $f_{n1,n4}$ and $f_{n3,n4}$ to calculate $\text{GNDF}_{G2}^{n4}=\frac{\text{GLDF}_2^{n1,n4}\cdot f_{n1,n4} + \text{GLDF}_2^{n3,n4}\cdot f_{n3,n4}}{ f_{n1,n4} +f_{n3,n4}}$. With similar logic, we can analyze the rest of the buses in the topological order by iteratively updating these two matrices.

\begin{figure}[tb]
    \centering
    \includegraphics[width=0.95\linewidth]{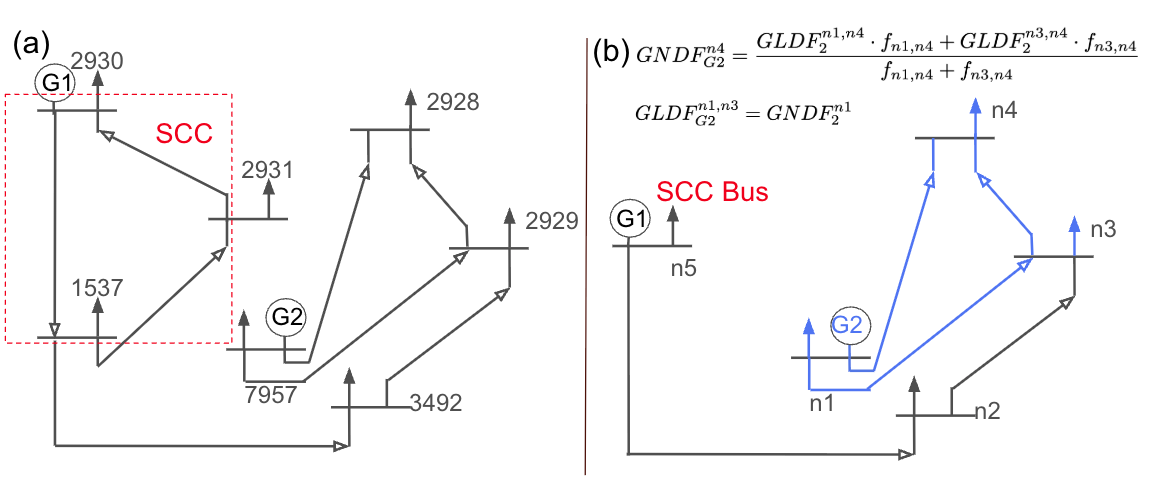}
    \caption{Examples on dealing with strongly connected components (SCC) and carbon tracing process. (a). A SCC exists in the power flow obtained after solving ACOPF in the realistic CAISO network.  (b). Once SCC is equivalently aggregated as one super node $n5$ following approach in \ref{sec:carbon_tracing2}, our proposed approach proceeds by computing the $\text{GNDF}$ and $\text{GLDF}$ matrices following approach in \ref{sec:carbon_tracing1}. Generation $G2$'s contribution to bus $n4$ and line flow $\{n1,n3\}$ are taken as examples. }
    \label{fig:SCC}
\end{figure}

To initialize this iterative process, at the initialization stage, we need to identify the load buses reachable by each generator. To achieve so, we treat each generator $k, k=1,..., K$ as the root node, and use Depth First Search (DFS) to get the buses and lines reachable from this $k$-th generator, which creates a mapping $\mathcal{M}(j)$ from bus $j$ to the generator set. In the mapping, the key is the bus index $j$, and each value is a set of generators supporting this bus's power demand. In the same DFS recursion, we can also get the backtracking map. To initialize the $\text{GNDF}$ and $\text{GLDF}$ matrices and reduce complexity in subsequent calculations for each matrix entry, we flag $\text{GNDF}_k^j$ to 1 if bus $j$ is directly reachable by generator $k$ in the directed graph, and bus $j$ has no other inflow lines, e.g., $\mathcal{L}_{in}^j=\emptyset$. Similarly, $\text{GLDF}_k^l=1$ if line $l$ is flowing out of generator $k$, and $\mathcal{L}_{in}^k=\emptyset$. To initiate the iterative process, we need to store buses with zero in-degrees in an initial queue $\mathcal{Q}$. 
By doing so, in the latter stages we pop buses out of the queue in topological order and process them accordingly.

For each node $i \in \mathcal{N}$, we also compute and store the in-degrees $\text{In\_Deg}(i)\geq 0$ beforehand. In order to identify the starting point for calculating each entry in $\text{GNDF}$ matrix, we utilize Kahn's algorithm~\citep{kahn1962topological} for topological sorting. It is a classic algorithm for deciding a valid ordering of tasks with dependencies on the directed acyclic graph, and is capable of processing very large networks. The process starts with  nodes of $\text{In\_Deg}(\cdot)=0$ first, and remove the outflow lines of these nodes. The iterative process based on topological ordering is illustrated in Algorithm \ref{alg:cap}.
\begin{algorithm}
\caption{Carbon Tracing for Directed Graph}\label{alg:cap}
\begin{algorithmic}[1]
\STATE Initialize: $M_1$ as mapping from nodes to generators
\STATE Initialize:  $\text{GNDF} \in \mathbb{R}^{K \times N}, \text{GLDF} \in \mathbb{R}^{K \times L}$ as matrices of zeros
\STATE Initialize: $\mathcal{Q} \gets $ a queue of nodes with $\text{In\_Deg}(\cdot) = 0$
\STATE $\bm{\gamma} \gets $ Generation emission rates vector
\STATE $\{\text{In\_Deg}(i), i \in \mathcal{N}\} \gets$ a dictionary of bus in-degree
\STATE $\delta_j \gets$ node $j$'s average emission rate $\forall j$ 
\STATE \# Run for all nodes in the graph
\WHILE{$\mathcal{Q} \neq \emptyset$}
\STATE $j \gets$ pop($\mathcal{Q}$) 
\FOR{
$k \in \mathcal{M}(j)$
}
\STATE $\text{GNDF}_k^j \gets \frac{\sum_{l \in \mathcal{L}_{in}^j} \text{GLDF}_k^l\cdot f_l}{\sum_{l\in \mathcal{L}_{in}^j}f_l} $
\FORALL{$l \in \mathcal{L}_{out}^j$}
\STATE     $\text{GLDF}_k^l \gets \text{GNDF}_k^j$
\ENDFOR
\ENDFOR
\FOR{$b_j,v \in \mathcal{L}$ }
\STATE $\text{In\_Deg}(v) \gets \text{In\_Deg}(v) -1$
\IF{$\text{In\_Deg}(v) == 0$}
\STATE $\mathcal{Q} \gets $ add $v$
\ENDIF
\ENDFOR
\ENDWHILE
\RETURN $\delta_j \gets [\mathbf{GNDF}^j]^T \cdot  \bm{\gamma}$ 
\end{algorithmic}
\end{algorithm}

The algorithm stops when there is no more bus in the queue $\mathcal{Q}$, i.e. there is no bus with $\text{In\_Deg}=0$ anymore. 
Kahn's algorithm stops until it processes all the nodes in a directed acyclic graph (DAG). All the buses proportion values and line proportion values in our matrices are calculated as long as our graph topology is DAG. In the DCOPF case, previous works show that no cycle flow exist~\citep{chen2024contributions, chen2022learning}, making our algorithm directly applicable. While in the solutions given by ACOPF, there may exist directed cycle flows, where there exists a non-empty directed trail in which the first and last nodes are equal. To make our carbon tracing algorithms apply to this general setting, we design a cycle detection and elimination framework in the next subsection.  

\subsection{Identifying and Merging Strongly Connected Components}
\label{sec:carbon_tracing2}
To tackle the issue of cycle flows in ACOPF solutions, we adopt the notion of Strongly Connected Components (SCC) in graph theory to transform the original directed graph with SCC(s) into a graph without SCC.

SCC is a maximal set of load buses such that any two buses of this subset are reachable from each other. If a set of buses are in a directed cycle, then they are all reachable from each other, so they are also inside an SCC by definition. Moreover, a directed cycle is a single directed path, while SCC gives the sub-graph where there is a directed path from every vertex in the sub-graph to every other vertex in the sub-graph. Mathematically, let $\text{SCC}_k\in \mathcal{S}$ denote the SCC $k$ inside the whole set of SCC $\mathcal{S}$ given the directed graph $G^{\text{OPF}}$, for any pair of buses $\{i,j\} \in \text{SCC}_k$, $ i \rightarrow j, j \rightarrow i$ denote they are reachable from both directions. 
Subsequently, we apply Tarjan's algorithm~\citep{tarjan1972depth} to identify the strongly connected components (SCCs) within the directed graph. 
Once all SCCs are identified and collapsed into super nodes, the overall graph transforms into a DAG. 

\begin{lemma}
    In the directed graph $G^{\text{OPF}}$, for any two SCCs $\text{SCC}_i, \text{SCC}_j \in \mathcal{S}$, they do not intersect.
\end{lemma}

See proof in Appendix \ref{sec:App_Proof} of our online preprint~\citep{shen2025carbon}. Since SCCs do not intersect each other, this guarantees once SCCs are all identified, all buses in the power network are separated into mutually exclusive groups, either as a member of exclusive, non-intersecting SCC or as a standard node. Thus, we define $\tilde{G}^{\text{OPF}}$ as a DAG, where every SCC is grouped and merged as a super bus. In $\tilde{G}^{\text{OPF}}$, each newly-constructed bus $\tilde{i}$ corresponds to the $\text{SCC}_i$ of $G^{\text{OPF}}$. For nodes residing outside of SCC in $G^{\text{OPF}}$ and connected to the original SCC, they are now directly connected to the merged new node. For the newly indexed bus $\tilde{i}$, we sum over all generations and loads across the original $\text{SCC}_i$ and are denoted as $p_{\tilde{i}}^g$ and $p_{\tilde{i}}^d$ respectively. The following Proposition guarantees that after the transformation from $G^{\text{OPF}}$ $\tilde{G}^{\text{OPF}}$, the solutions of carbon tracing for nodes outside SCCs are not altered:

\begin{proposition}
    Consider the original ACOPF solution graph $G^{\text{OPF}}$ and transformed DAG $\tilde{G}^{\text{OPF}}$. Let $\delta(p_i^d|G^{\text{OPF}})$ and $\delta(p_i^d|\tilde{G}^{\text{OPF}})$ denote node $i$'s average carbon emission rate under graphs $G^{\text{OPF}}$ and $\tilde{G}^{\text{OPF}}$, respectively. For all nodes that are not included in SCCs, $\forall i \notin \mathcal{S}, \delta(p_i^d|G^{\text{OPF}})=\delta(p_i^d|\tilde{G}^{\text{OPF}})$ under all load conditions $\mathbf{p}^d$ and corresponding $\mathbf{p}^g$.
\end{proposition}

\begin{proof}
    Proof of the Proposition is straightforward. To show the resulting emission rates of $i \notin \mathcal{S}$ are not affected by SCC transformation, we firstly consider the case where $i$ is an upstream to a neighboring $\text{SCC}_j$, e.g., $i\rightarrow \text{SCC}_j$. After transformation of SCC into super nodes, since there is no change of power or carbon emission flowing into bus $i$, $\text{GLDF}_k^{i, \text{SCC}_j}$ does not change for any $k\in K$. Then emission profile for bus $i$ keeps unchanged. 
    Following the similar logic, we can show emission rate is unchanged when $i$ is an downstream node to a neighboring SCC, $\text{SCC}_j \rightarrow i$.

    For all other nodes not directly connected to $SCC_j$, we can prove by induction that their emission rates also remain unaffected by the SCC transformation, e.g., by showing that the degree-2's neighbor nodes (nodes which are two directed edges away) of $SCC_j$ do not experience any changes in their emission profiles. 
\end{proof}

Once we identify the SCCs, we merge all the nodes in the same SCC into a super node where its power demand is the sum of all the original nodes in the directed cycle. After implementing Tarjan's algorithm, we have the updated graph topology $\tilde{G}^{\text{OPF}}$, which is \emph{acyclic}. 
To prove by contradiction, let us assume $\tilde{G}^{\text{OPF}}$ is cyclic, which means that there exist node $a \rightarrow b$ and node $b \rightarrow a$. The node $a$ corresponds to an SCC set $\mathcal{A}$ in $\tilde{G}^{\text{OPF}}$
and the node $b$ corresponds to an SCC set $\mathcal{B}$ in $\tilde{G}^{\text{OPF}}$. So $\forall$ bus $x \in \mathcal{A}$, and $\forall$ bus $y\in b$, $x \rightarrow y$ and $y \rightarrow x$. This means $\mathcal{A}$ and $\mathcal{B}$ should belong to the same SCC instead of two separate SCCs. 

We note that previous flow-based approaches for emission tracing only discuss the settings where there exist no cycles, or work under the assumption of a linear relationship between power generation and power injections to line flows and nodes~\citep{kang2015carbon, shaocarbon, MISO2025}. Further, previous works need to fully calculate a generator-to-load dispatch ratio matrix and take the pseudoinverse with the assumption that such a matrix is invertible. Our SCC identification step is an order of magnitude more efficient and also matrix-inversion-free. More importantly, the resulting algorithm is compatible with general, nonlinear power flow models and OPF solutions because of the design on cycle flows detection and elimination techniques.





\subsection{Calculating Locational Emission Rates}
\label{sec:emission_rates}
Given the compete \text{GNDF} and \text{GLDF} matrices computed from Algorithm \ref{alg:cap}, carbon emission rates $\gamma_k$ for all generators $k \in {1, ..., K}$, the nodal average emission rate-LAE at bus $i$ can be computed as 
\begin{equation}
\label{equ:nodal_emission}
    \delta(p_i^d)=\sum_{k=1}^{K} \gamma_k \cdot \text{GNDF}_k^i=[\mathbf{GNDF}^i]^T\cdot\bm{\gamma}.
\end{equation}
By collecting the nodal emission rate for all buses, we derive the locational carbon emission rate vector as
\begin{equation}
    \boldsymbol{\delta} = \begin{bmatrix}
                \delta_{1} \\
                \delta_{2} \\
                \vdots \\
                \delta_{N}
            \end{bmatrix}	= \begin{bmatrix}
                \text{GNDF}_1^1 & \dots & \text{GNDF}_K^1 \\
                \vdots &  \ddots   & \vdots \\       
                \text{GNDF}_1^N & \dots & \text{GNDF}_K^N
            \end{bmatrix}
            \begin{bmatrix}
                \gamma_{1} \\
                \gamma_{2} \\
                \vdots \\
                \gamma_{K}
            \end{bmatrix} .
\end{equation}

Our algorithm can be also easily extended to calculate LMEs. To achieve that, denote $\mathbf{\hat{p}}^g$ as the generation dispatch vector corresponding to the locally perturbed load $\hat{p}_i^d$, which can be acquired by resolving the OPF \eqref{equ:opf} with the perturbed load. Then the LMEs at node $i$ can be calculated as 

\begin{equation}
    \mu(p_i^d)=\frac{\sum_{k=1}^{K} (\gamma_k \cdot \hat{p}_k^g- \gamma_k \cdot p_k^g)  }{\hat{p}_i^d- p_i^d}.
\end{equation}

\subsection{Time Complexity}
For our model, $|\mathcal{N}|$ is the number of buses and $|\mathcal{L}|$ is the number of transmission lines.  
The time complexity of finding SCCs is $O(|\mathcal{N}|+|\mathcal{L}|)$, because it performs a DFS that visits every node of the graph exactly once, and it does not require to revisit any node that has already been visited. For the topological sorting, Kahn's algorithm uses stacks to keep track of nodes. Stack operations (push and pop) take $O(1)$ time per operation. Each node is pushed onto the stack once and popped once. Therefore, the time complexity of Kahn's algorithm is linear in the number of nodes $(|\mathcal{N}|)$ and edges $(|\mathcal{L}|)$ in the graph: $O(|\mathcal{N}|+|\mathcal{L}|)$.

The time complexity of topologically sorting the matrices is also $O(|\mathcal{N}|+|\mathcal{L}|)$, and the time complexity of finding and attributing SCCs in graph is $O(|\mathcal{N}|+|\mathcal{L}|)$. So the overall time complexity of our whole algorithm remains $O(|\mathcal{N}|+|\mathcal{L}|)$. This time complexity is a significant improvement over the depth-first tree search based method applied in previous research~\citep{chen2024contributions}, which is $O(|\mathcal{N}||\mathcal{L}|)$. While for methods relying on matrix-inversion and linearized power flow~\citep{chen2024carbon, kang2015carbon, lu2024market}, the complexity is in the order of $O(|\mathcal{N}|^3+|\mathcal{N}||\mathcal{L}|)$.

We also time the running process of algorithm in IEEE systems for carbon emission calculation. The average runtime of our carbon tracing algorithm in an IEEE 6-bus case study is 0.0269 seconds per instance. And the average runtime in IEEE 30-bus study is 0.07674 seconds per instance, adding very mild computations compared to realistic OPF solving process.

\section{California Grid Case Study Setup}
In this section, we detail the setup on testing CAISO network's nodal emission. To our knowledge, this is one of the first systematic nodal-level carbon emission studies working with real-world data and realistic transmission topology. Implications brought by this work can strengthen the understanding of fine-grained grid emissions, which provides necessary tools and simulation testbeds for carbon emissions.

\subsection{Network Topology}
In order to test the accuracy and efficiency of our carbon tracing algorithm, we utilize the California Test System (CATS)  \cite{CATS2024}, which is a geographically-accurate, synthetic electric grid model without CEII-protected information. The model provides geographic coordinates of the loads and generators, transmission line and transformer parameters, and manually curated grid topology. 
Fig. \ref{fig:large}(a) visualizes the CATS model topology, where  transmission line thickness represents voltage level, and the size of the blue dots represents bus load demand. The test power network has 8,870 buses, 2,149 generators, and 10,162 transmission lines. The CATS system is with about 73 GW of existing generation capacity, and 1 GWh of existing energy storage capacity.  Compared to IEEE test cases usually with limited number of buses and synthetic load data, due to the nature of CAISO's large-scale network and CATS realistic loading conditions, benchmarking on this model can provide a comprehensive evaluation of our carbon tracing algorithm.\vspace{-10pt} 

\subsection{Demand, Generation, and Emission Data}
\paragraph{Nodal Demand Data Estimation} To get the real-time power demand data, we query the CAISO API~\cite{CAISO_demand} to collect and process the 2019 whole-year data. The original data is in 5 minutes resolution, and we down-sample to hourly demand data for the carbon rate computation. Since CAISO's published data only include system-level demand, we make the assumption that the load demand distribution remains relatively static. We follow a similar load assignment strategy as the method developed in \cite{CATS2024}, where the distribution of demand across nodes is determined by historical consumption patterns and demographic factors. Based on CATS's estimated nodal load vector, we scale the nodal demand at each hour to match the {hourly} aggregated CAISO system demand data.

\paragraph{Nodal Renewable Generation Estimation} In both CATS and CAISO systems, there are no regional or nodal renewable generation values available. To implement our carbon tracing algorithm, we follow the steps described in Appendix \ref{sec:App_renewable} to distribute the total CAISO system-level renewable generation to each node based on an optimization-based procedure. The carbon emission rate for renewable generation (i.e., $\gamma_k$) is set to be zero in the case study. 

\paragraph{Imported Power and Emission Estimation} Since neither CATS model nor CAISO reveals the import energy price, the power generation imported from other states does not have an estimated cost function which can be applied to solve OPF. To address that, we propose to allocate the total import value published by CAISO to each state-bordering node in CATS, and make sure nodal power balance is satisfied after such separation. Once we treat the bordering buses as virtual power generation node with generation equal to import power, we do an estimation for import power's emission rate. As the total emissions from import can be calculated by subtracting emissions from other generation sources. We then fit the import's emission rate based on historical $\{{p}^{\text{import}}, e^{\text{import}}\}$ data pairs, where ${p}^{\text{import}}$ is the imported power (in MW) and $e^{\text{import}}$ is the imported emission (in metric ton of CO2).

\paragraph{Emission Rate Data} 

Table \ref{table:emission_rate} in Appendix shows the generator emission rate $\gamma_k$ data we used for different types of generators~\cite{EIA1}. Furthermore, 
CAISO API~\cite{CAISO_carbon} provides the hourly total emission by each energy source in California at 5 minutes resolution. We treat these aggregated emission values as a reference to validate the locational emission rates and emission per energy resources estimated from our algorithm.  


All of our simulations are implemented on a laptop with an Intel® Core™ Ultra 5 Processor 125H CPU. We find that all the optimal power flow problems based on the CATS topology, system and generator parameters, and our nodal demand estimation are feasible. For each simulation, we record the solved power generation and line flows and use these values to compute the nodal and total emissions to evaluate our method. During the year-long simulations, the average solution time for DCOPF is 2 minutes and 23 seconds, with a standard deviation of 9.51 seconds. ACOPF takes 4 hours and 26 minutes to solve, with a standard deviation of 1 hour and 13 minutes. Throughout the simulations, we find that DCOPF’s solution for generation dispatch is within $3\%$ of ACOPF, suggesting that DCOPF provides a good approximation of the full ACOPF in our simulated cases.

\section{Case Study Results}
\subsection{Validation of Model Fidelity}
To ensure our nodal emission calculations are accurate, we validate our OPF-based power generation results by summing each generator's output by fuel type and comparing them with CAISO’s published data; a randomly selected day’s results are shown in Fig. \ref{fig:generation_mix}. This procedure ensures that our estimated power dispatch $\mathbf{p}^g$ from OPF solutions using CATS model reflects actual system dispatch decisions, so our carbon tracing method operates on legitimate power flows $\mathbf{p}_l$. 
Note that power generation from renewable energy sources (Solar, Wind, Geothermal, Biomass, Biogas, and Hydro) is taken directly from ground-truth CAISO data rather than solved through our OPF model. This is because the bidding strategies of these sources are not explicitly available from CAISO or the CATS model, so we do not incorporate a cost vector for them. Moreover, these are stochastic generation sources, and predicting or estimating their output is beyond the scope of this work. Future power grid emission studies may explore how factors such as weather and geolocation affect short-term availabilities of renewable and hydroelectric generation.
\begin{figure}[!htb]
\centering
{\includegraphics[width=0.99\linewidth]{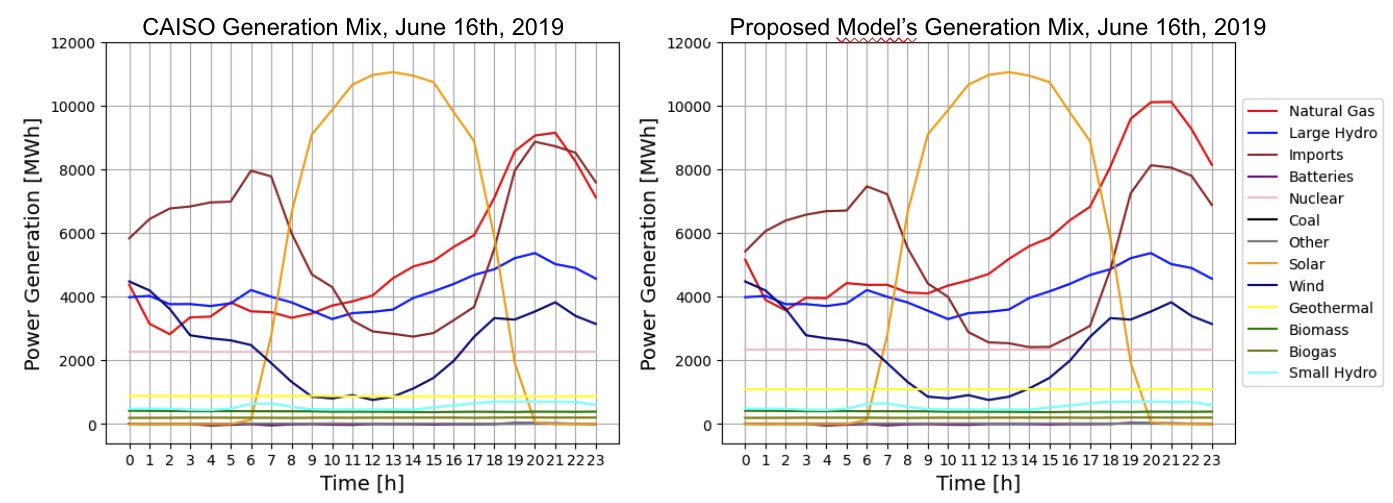}}
\caption{Power generation comparison by fuel type.}
\label{fig:generation_mix}
\end{figure}

\begin{figure}[!htb]
\centering
{\includegraphics[width=0.99\linewidth]{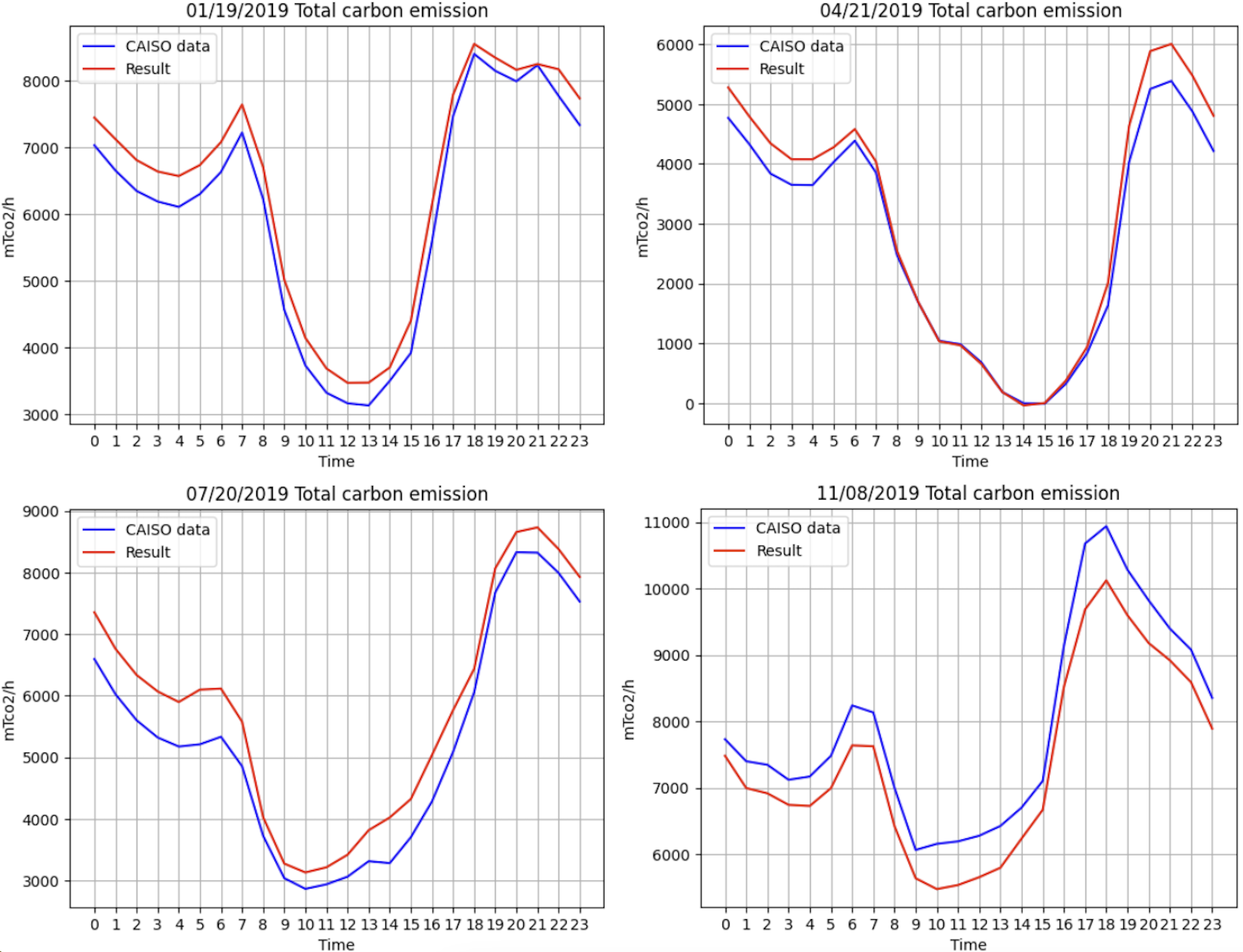}}
\caption{Validation of system-wide aggregate carbon emission for four randomly selected days across four seasons.}
\label{fig:time_series_carbon}
\end{figure}

Once we derive the GNDF matrix from our carbon tracing algorithm, at each timestep $t$ we can also estimate the system-level total emissions $E$ in metric tonnes of CO$_2$ as 
\begin{equation}
    E(t)=\sum \delta(p_i^d(t))\cdot p_i^d(t)= \sum_i^N ([\mathbf{GNDF}^i]^T\cdot \bm{\gamma} )\cdot p_i^d (t).
\end{equation}

The daily timeseries comparisons of our estimated $E(t)$ versus CAISO's reported system-level emissions are shown in Fig. \ref{fig:time_series_carbon}, for four randomly selected days in different seasons. We observe that our estimated $E(t)$—based on our proposed nodal power demand estimation and nodal emission rates computation method—closely tracks the CAISO published data on a daily basis. These results indicate that our proposed method can effectively model carbon emission rates first at the nodal level and then aggregate them at the system level.

The Mean absolute percentage error (MAPE) and weighted Mean absolute percentage error (wMAPE) are used to measure the accuracy of our total power generation and total carbon emission results. The wMAPE is averaged upon each data entry's relative error rate. In Table \ref{table-fourseason}, we record monthly emission rate and error rate compared to CAISO's published values for four typical month. Winter months generally experience higher emission rate due to scarcity of solar generation in the winter. Across all four seasons, our calculated emission rate are within $12.2\%$, indicating the proposed approach can closely track both the generation dispatch and emission patterns. Details of the evaluation metrics and additional validation results across yearly data are described in Appendix. 

\begin{table}
\caption{Average carbon emission rate (metric tonnes of CO$_2$/MWh) based on the estimated total system-level emission v.s. CAISO published values.}
\label{table-fourseason}
\begin{tabularx}{0.49\textwidth} {b|ssss}
\Xhline{2\arrayrulewidth}
Months  
& Mar & Jun & Sep & Dec  \\ 
\midrule
Calculated Rate & 0.211 &0.2104 &0.2494 &0.2875\\
CAISO Rate & 0.2255& 0.2189 &0.2862 &0.318\\
MAPE (\%)   & 9.86& 2.44&12.17&11.25  \\ 
wMAPE (\%)   & 8.61 &2.28&11.66&10.78 \\
\bottomrule
\end{tabularx}
\end{table}

\subsection{Temporal and Spatial Variations of Emission Rates}
\begin{figure}[htb]
\centering
{\includegraphics[width=0.99\linewidth]{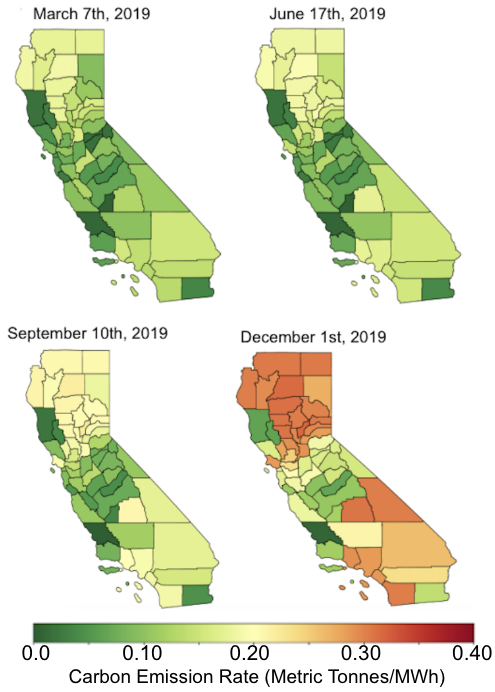}}
\caption{Geographical illustration of locational average emission (LAE) rates at 11 am for four seasons, with one day randomly selected from each season.\vspace{-10pt}}
\label{fig:timeseries}
\end{figure}
In Fig. \ref{fig:timeseries}, we visualize the carbon emission rates at the county level over four different seasons. By randomly selecting one date from each season, we observe that the winter season is characterized by the most geographical variation in emission rates, while inland counties generally possess higher locational emission rates. This is attributed to the industrial activities concentrated in these areas, along with lower access to cleaner energy resources. Conversely, coastal regions, benefiting from a greater influx of renewable energy sources such as wind and solar, exhibit comparatively lower emissions.

\begin{figure}[tb]
\centering
{\includegraphics[width=0.99\linewidth]{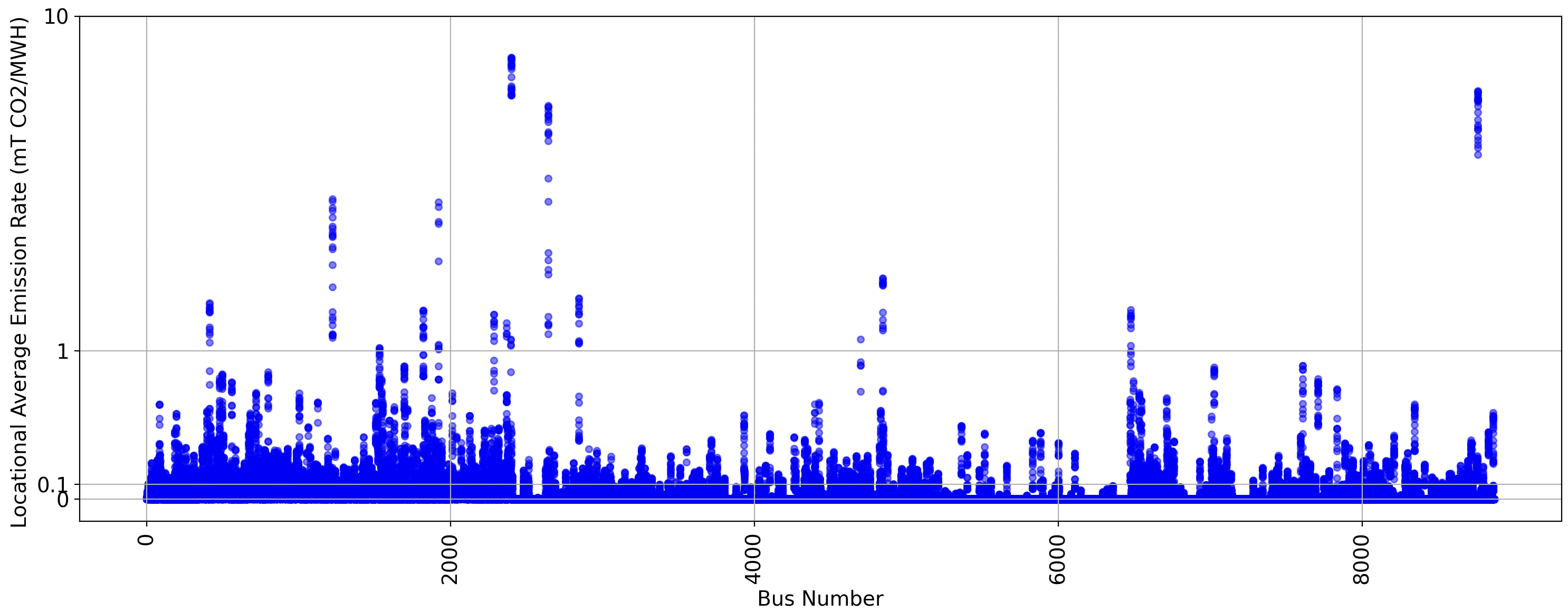}
}
\caption{Distribution of 24-hour's LAE across 8,870 nodes in the simulated CAISO network on June 17th, 2019.\vspace{-10pt}}
\label{fig:nodal_distribution}
\end{figure}

It is noteworthy that San Luis Obispo County consistently has the lowest emission rate across all seasons, as the state's largest single power station—the Diablo Canyon nuclear power plant, with a nameplate capacity of 2,256 MW—is continuously generating carbon-free electricity. Thus, our carbon flow tracing algorithm always analyzes the outflow of low-carbon electricity from this county, marked in dark green in Fig. \ref{fig:timeseries}. Since the Diablo Canyon power plant is scheduled to retire in a few years, California's geographical distribution of emission rates may change significantly. Fig. \ref{fig:nodal_distribution} illustrates how the locational average emission rate vary across all nodes in the 24 hours of June 17th. Each column in the figure represents the occurrences of certain emission rate at a specific bus.  There are few outliers with much higher emission rates ($>0.5$ metric tonnes CO$_2$/MWh), indicating the very infrequent dispatch by coal-fired generators in the CAISO network. Notably, regions with higher renewable energy generation often coincide with lower emissions, indicating the local effects of integrating sustainable practices. Coastal regions typically have lower LAE rates due to the more diverse energy mix.

\begin{figure}[tb]
\centering
{\includegraphics[width=0.99\linewidth]{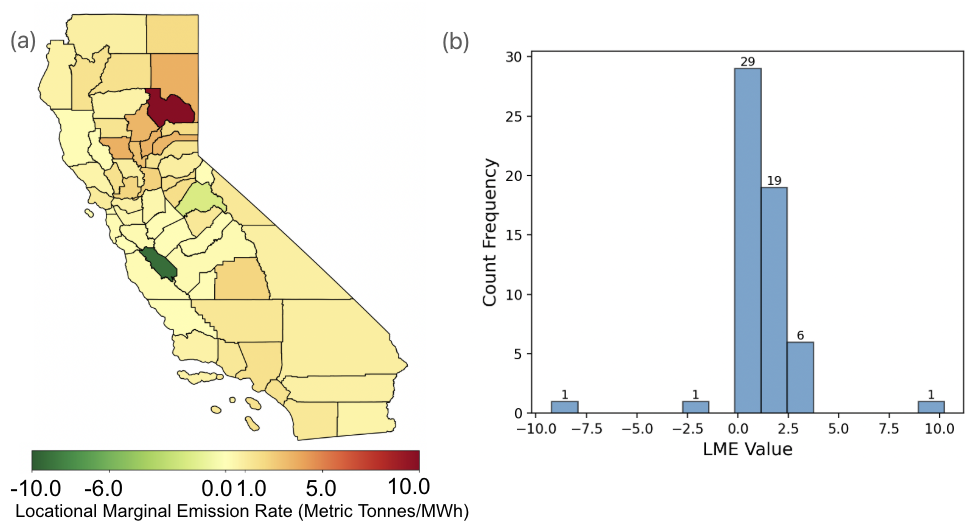}
}
\caption{(a) Geographical illustration of county-level LME; (b) Distribution of LME rates for all counties in California. Results are shown for 11am, June 17th, 2019.\vspace{-20pt}}
\label{fig:LME}
\end{figure}

In Fig. \ref{fig:LME} we illustrate the geographical variations of LME at 11am on June 17th, 2019. The colorbar is coded based on double-sided logscale axis. We can observe the LME rates vary significantly across the state. Plumas County in Northeastern California records the highest LME rate of $10.21 $ metric tonnes CO$_2$/MWh. This may be due to the relatively low power consumption with small population of $18,807$ in Plumas County, so that the power dispatch is varying a lot depending on loading conditions. LME can also drop below 0, as increasing the total load may lead to more renewables being dispatched, causing the negative marginal emission rate. Such a spatial heterogeneity also underscores the importance of energy management strategies considering distributed energy resources, and provide localized insights for emission reduction initiatives. More results are in Appendix \ref{sec:more_results}.



\section{Conclusion and Discussion}
In this work, we propose a systematic framework for computing locational carbon emission rates through power flow analysis. Our method supports full ACOPF formulations, enabled by cycle flow identification and elimination techniques. The algorithm achieves high computational efficiency, making it suitable for large-scale, realistic power system applications. We apply this framework to perform an ISO-scale, hourly emission analysis on the California grid. 
Results underscore the need for geographically informed policy and infrastructure planning to support effective and equitable decarbonization.

Future work can incorporate demand forecasts and generation prediction for emission estimation. 
There are instances where the network model might not be available to the interested party. An intriguing research area involves understanding how to estimate the topology, parameters, and emission profiles of the electricity network using publicly available data.

\bibliographystyle{plainnat}
\bibliography{references}

\beginappendix

\appendix
\section{Proof of Proposition 1}
\label{sec:App_Proof}
\begin{proposition}
    Sets of SCCs do not intersect each other.
\end{proposition}
\begin{proof}
    We prove by contradiction. Assume $\mathcal{S}_{1}$ and $\mathcal{S}_{2}$ are two SCCs in the given directed graph. If \(\mathcal{S}_1\) and \(\mathcal{S}_2\) intersect, i.e. there exists bus $a \in \mathcal{S}_1$ and $a \in \mathcal{S}_2$. Since $\mathcal{S}_1$ is a SCC, $a \rightarrow \mathcal{S}_1 \backslash \{a\}$ and $\mathcal{S}_1 \backslash \{a\} \rightarrow a$. Since $\mathcal{S}_2$ is a SCC, $a \rightarrow \mathcal{S}_2 \backslash \{a\}$ and $\mathcal{S}_2 \backslash \{a\} \rightarrow a$. 
    
    Because reachability is transitive, $\mathcal{S}_1 \backslash \{a\} \rightarrow a \rightarrow S_2\backslash\{a\} \Rightarrow \mathcal{S}_1\backslash\{a\} \rightarrow \mathcal{S}_2\backslash\{a\}$. Similarly, $\mathcal{S}_2\backslash\{a\} \rightarrow \mathcal{S}_1\backslash\{a\}$. So $\mathcal{S}_1$ and $\mathcal{S}_2$ can merge together to a larger set of SCC, denoted as $\mathcal{S}_m$. 
    
    Since $\mathcal{S}_1 \in \mathcal{S}_m$ and $\mathcal{S}_2 \in \mathcal{S}_m$, $|\mathcal{S}_m| > |\mathcal{S}_1|$ and $|\mathcal{S}_m| > |\mathcal{S}_2|$. $\mathcal{S}_1$ and $\mathcal{S}_2$ are not maximal set of load buses that are any two buses are reachable, i.e. they are not SCCs by definition. So this is proved by contradiction.
\end{proof}

\section{Data Processing for \\Renewable Generation and Imports}
\label{sec:App_renewable}

As CAISO only publishes aggregated demand and net imports, while in order to calculate directed power flow and associated emission flows, we need to know each node's renewable power generation. Another challenge is that in the  original value in the Matpower file provided by the CATS model~\citep{CATS2024}, import values is not recorded and they  will be 0 if directly loaded into OPF solution process. Moreover, there is no cost vector for renewables, making it hard to directly find the market clearing results. To avoid that, we develop a tailored approach to estimate nodal renewable generation based on the total renewables generation published by CAISO. 

Let $\mathbf{p}_{total}^d$ represent total load vector across all buses under which all types of generators are considered. Let $\mathbf{p}_{total}^g$ represent total power generation of all types of generators, and we have $\sum p_{total}^d = \sum p_{total}^g$ always hold. Under the case where renewable generation is not taken into account, we use $\mathbf{p}_{net}^g$ to represent power generation of non-renewable generators and $\mathbf{p}_{net}^d$ represent the net load of the buses after excluding renewable generation. We are interested in finding each node's renewable contribution $p_{i}^d(p_r^g)$, with $r\in \mathcal{R}$ denoting the index for renewable generator. We then have  
\begin{equation}
    p_{i, total}^d = \sum_k p_{i}^d(p_{k,net}^g)+ p_{i}^d(p_r^g).
\end{equation}

To properly assign $p_{i}^d(p_{k,net}^g)$, we propose to solve the OPF model with only the non-renewable generation components, using CATS's model's published cost functions for each generator. This approach allows us to obtain a baseline scenario against which we can compare the impacts of integrating renewable generation into the network. The optimization problem will focus on minimizing operational costs while satisfying load demand and other operational constraints as listed in \eqref{equ:opf}, and we use $OPF(\cdot)$ to denote such a solution process and the mapping function from load to corresponding generation:
\begin{equation}
    \mathbf{p}_{net}^g = OPF(\mathbf{p}_{net}^d);
\end{equation}
After finding the directed power flow based on $\mathbf{p}_{net}^g$ and then determining each non-renewable generator's nodal power contribution,  we are ready to find each node's renewable generation level while satisfying power balance constraints:
\begin{equation}
   p_{i}^d(p_r^g)  = p^d_i - \sum_k p_{i}^d(p_{k,net}^g);
\end{equation}

After getting the proportion of power supplied by renewable generators at each bus, we can also calculate the carbon intensity of the two scenarios with or without considering renewables dispatch. The results shown across this paper are for the scenario with renewables dispatch, which takes the renewables' emission reduction potentials into account.



\section{More Results on California Grid}
\label{sec:more_results}
\subsection{Evaluation Metrics:}
We adopt MAPE and wMAPE for evaluating the algorithm performance. We take the evaluation of power generation as an example here. Within evaluation time horizon, $p_{t}^{cg}$ is the actual power generation data from CAISO and $p_{t}^{og}$ is the power generation result from the OPF model.
\begin{equation}
\label{equ:metrics}
    MAPE = \dfrac{1}{T} \sum_{t=1}^{T} |\dfrac{p_{t}^{cg}-p_{t}^{og}}{p_{t}^{cg}}|; 
    wMAPE = \dfrac{\sum_{t=1}^{T} |p_{t}^{cg}-p_{t}^{og}|}{\sum_{t=1}^{T} |p_{t}^{cg}|}.
\end{equation}

Similarly, we can replace $p$ with $\delta$ for the evaluation of carbon emission rate in \eqref{equ:metrics}. In Fig. \ref{fig:generation_mix} and Fig. \ref{fig:time_series_carbon}, we use these metrics to verify the effectiveness of our proposed approach on generation mix and system-level emissions respectively.


\subsection{Additional Results}
For each month, we randomly select one day as a test sample, and calculate the MAPE and wMAPE across 24 hours. The final results can be found in Table \ref{my-label}. It shows a higher error rate of emission rate calculate in May and June. This could be because the unstable nature of solar and hydro renewable generation during summer months. While for both power generation and emission calculation, our approach can closely track the ground truth value, indicating the our underlying OPF engine can well represent the power dispatch process for CAISO grid. This validation procedure makes sure proposed emission tracking algorithm can provide high fidelity results for locational emission analysis.

\begin{table*}
\caption{Carbon emission (C) and power generation (G) MAPE (\%) rate across 12 months. Ground truth values are from CAISO's published emission and generation data.}
\label{my-label}
\begin{tabularx}{\textwidth} {b|ssssssssssss}
\Xhline{2\arrayrulewidth}
Months  
& Jan & Feb & Mar & Apr & May 
& June & July & Aug & Sept & Oct & Nov & Dec \\ 
\midrule
MAPE(G)   & 3.00      & 6.39          & 8.00   & 7.35  & 13.19    & 11.64   & 8.46    & 6.78     & 7.47   & 4.85    & 3.62 & 1.97  \\ 
wMAPE(G) & 3.05        & 6.55        & 7.61    & 7.89   & 13.27    & 11.26   & 8.71    & 6.43     & 7.75   & 3.51    & 3.16 & 1.81  \\ 
MAPE(C)       & 6.84       & 3.35          & 4.29  & 8.90  & 12.23   & 13.37  & 11.40   & 5.10    & 5.18    & 13.29     & 7.03 & 2.56  \\ 
wMAPE(C)        & 6.08       & 2.58          & 4.09   & 10.16 & 12.79   & 12.09 & 10.48 & 5.16   & 4.52    & 12.09    & 6.93     & 2.71  \\ 
\bottomrule
\end{tabularx}
\end{table*}

In Fig. \ref{fig:distribution}, we plot and compare the distribution of locational average emission rate in three counties: Amador, Sacramento, and San Diego. A clear discrepancy can be observed in the emission profiles among these counties, highlighting the impact of local generations/transmission and geographical factors on emissions. Amador County, characterized by its rural landscape and lower population density, exhibits significantly more stochastic LAE rates compared to the more urbanized areas.
\begin{figure}[h]
    \centering
    \includegraphics[width=0.99\linewidth]{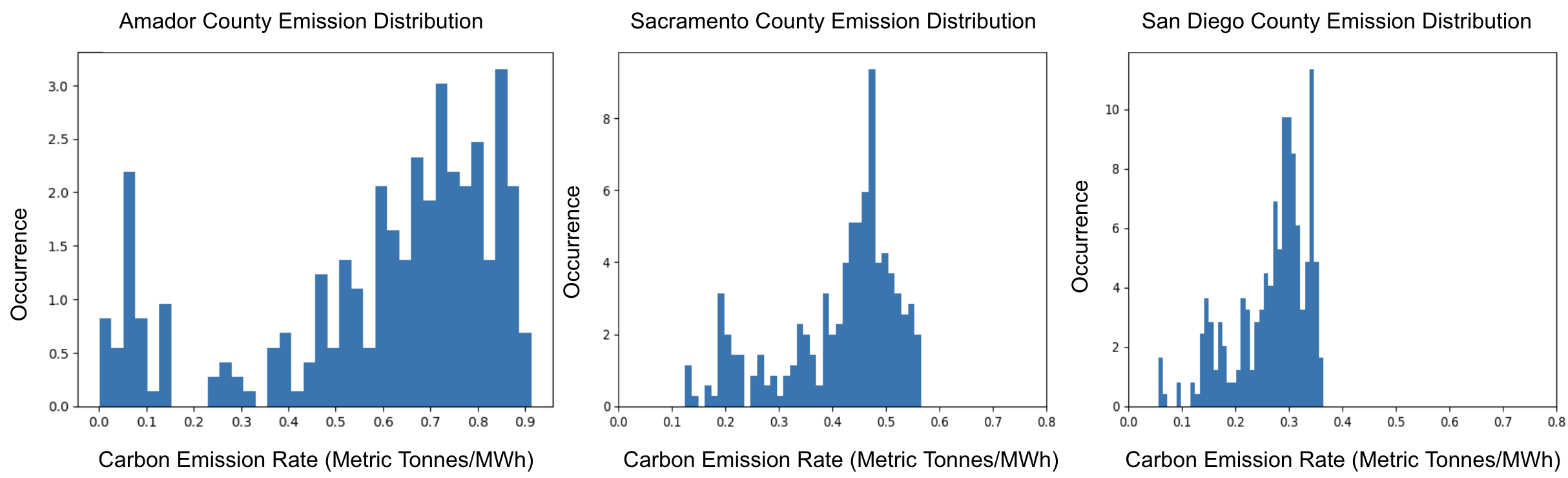}
    \caption{Comparison of distribution of Amador, Sacramento, and San Diego County's regional averal emission rate.}
    \label{fig:distribution}
\end{figure}

\begin{figure}[tb]
\centering
{\includegraphics[width=0.99\linewidth]{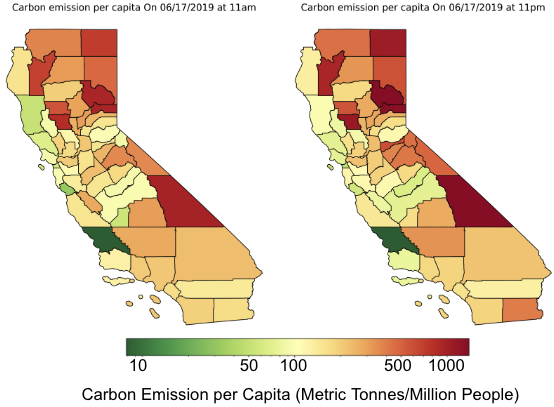}
}
\caption{Geographical illustration of carbon emissions per capita at 11 am and 11 pm on a random selected date.  }
\label{fig:captia}
\end{figure}
To identify the geographical and temporal disparities in emission profiles, Fig. \ref{fig:captia} visualizes the per-capita carbon intensity at 11 a.m. and 11 p.m. on a randomly selected date. The value is calculated by first identifying each county’s total emissions using the locational average emission rate and estimated demand, then dividing by the county’s population. 
We clearly observe that northern Californian counties do not show significant changes in total emissions between daytime and nighttime, mostly likely because  their renewable generation comes from hydro and wind rather than solar, exhibiting less daily variations. In contrast, solar contributions in southern counties are more significant during the daytime, leading to a larger difference in per-capita carbon intensity between day and night. In general, there is over 100 times difference in terms of emissions per capita associated with electricity usage, highlighting the geographical disparities. It is also notable that based on the results of aggregated emissions by county, Los Angeles County and San Bernardino County have the highest total emissions due to their large populations and high power demand.  Despite that, San Bernardino County exhibits a low per capita carbon emission rate because of a high share of renewable contribution, particularly from solar—thanks to its unique geography and some of the largest solar generation capacities in the state \citep{CEC1}. Interestingly, major metropolitan areas in California, such as the San Francisco Bay Area, Los Angeles, and San Diego, have high population concentrations. While they do not generate much renewable energy locally, they maintain a moderate per-capita emission profile by importing greener electricity from neighboring counties. Meanwhile, Inyo County and Imperial County have a high ratio of renewable generation to electricity use \citep{CEC2}.

In Fig. \ref{fig:emission_time}, we illustrate the geographical distribution of locational average emission rate across California at 11 a.m. and 11 p.m. on June 17th, 2019. The figure clearly shows the distinct patterns between daytime and nighttime. Interestingly, the inland counties are observing higher emission rate during night. Counties with higher power demand such as San Francisco, San Diego, and Los Angeles counties experience mild emission rate, partially because the energy mix at these counties are more diverse.

\begin{figure}[!htb]
\centering
\label{fig:time_series}
{\includegraphics[width=0.99\linewidth]{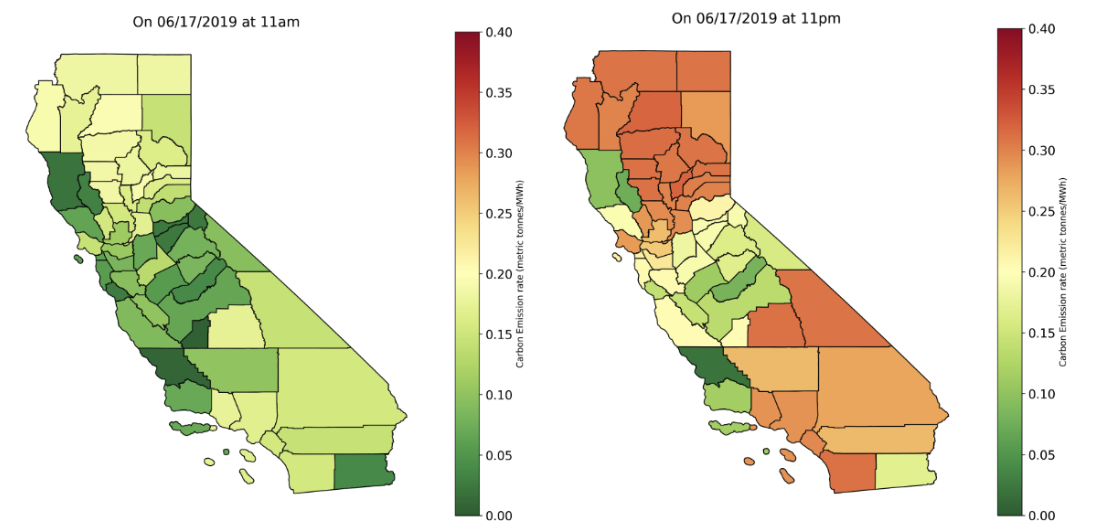}}
\caption{Geographical illustration of carbon emission rate at 11 am and 11 pm on a random selected date.}
\label{fig:emission_time}
\end{figure}

\section{Generator Emission Data}
In Table \ref{table:emission_rate}, we report the emission rate used in this study, which is based on the US Energy Information Administration (EIA)'s report on emissions by plant and by region~\citep{EIA1}. It is noteworthy currently in our simulation we use EIA's published, unified emission rate for Import. Future work can look into each specific import connection's emission rates, which can further improve carbon tracing results.
\begin{table}[h]
\small
\centering
\caption{Carbon emission rate by generator fuel type. t/MWh stands for metric tonnes per megawatt-hour.}
\label{table:emission_rate}
\begin{tabularx}{0.98\linewidth}{X|X X}
\Xhline{2\arrayrulewidth}
  Fuel Type & Emission Rate (t/MWh) & Category\\
  \hline
  Coal& 0.82 & Fossil\\ 
  Petroleum Liquids& 0.656 & Fossil\\ 
  Natural Gas& 0.44 & Fossil\\ 
  Nuclear& 0.0 & Low-Carbon \\
  Hydro & 0.0 & Renewable \\
  Biomass & 0.23 & Renewable\\
  Wind & 0.0 & Renewable \\
  Solar & 0.0  & Renewable\\
  Geothermal & 0.038 & Renewable \\
  Other/Import & 0.43 & Mix \\
\Xhline{2\arrayrulewidth}
\end{tabularx}
\end{table}

\section*{Acknowledgment}
Yuanyuan Shi and Yize Chen were supported in part by Climate Change AI Grant. Yize Chen was supported in part by  Natural Sciences and Engineering Research Council of Canada. The authors would like to thank the discussions with Feng Zhao, Tongxin Zheng and Izudin Lelic at ISO New England, and Deepjyoti Deka at Massachusetts Institute of Technology.

\end{document}